\documentclass[11 pt]{article}

\usepackage{amsmath}
\usepackage{amssymb}
\usepackage{amsthm}
\usepackage{ stmaryrd }
\usepackage{hyperref}
\usepackage{mdwtab}
\usepackage{multirow}
\usepackage{mathabx}
\usepackage{algorithm2e}

\newtheorem{theorem}{Theorem}[section]
\newtheorem{proposition}[theorem]{Proposition}
\newtheorem{corollary}[theorem]{Corollary}
\newtheorem{lemma}[theorem]{Lemma}

\newtheorem{claim}[theorem]{Claim}

\theoremstyle{definition}
\newtheorem{definition}[theorem]{Definition}

\newcommand{\ria}{\rightarrow}

\newcommand{\NN}{\mathbb{N}}
\newcommand{\QQ}{\mathbb{Q}}

\newcommand{\FG}{\mathfrak{G}}
\newcommand{\lra}{\leftrightarrow}
\newcommand{\Lra}{\Leftrightarrow}

\newcommand{\denote}[1]{\llbracket#1\rrbracket}
\newcommand{\angled}[1]{\langle#1\rangle}
\newcommand{\sub}{\subseteq}
\newcommand{\bs}[1]{\boldsymbol{#1}}
\newcommand{\algproblem}[3]{\begin{tabular}{l l}
    \multicolumn{2}{l}{{{\sc #1}}}\\
    \hline
    
    Input:&\parbox{9.5cm}{#2}\\[0.2cm]
    
    Output:&\parbox{9.5cm}{#3}\\
  \end{tabular}}
  
\newbox\gnBoxA
\newdimen\gnCornerHgt
\setbox\gnBoxA=\hbox{$\ulcorner$}
\global\gnCornerHgt=\ht\gnBoxA
\newdimen\gnArgHgt
\def\Godelnum #1{%
\setbox\gnBoxA=\hbox{$#1$}%
\gnArgHgt=\ht\gnBoxA%
\ifnum     \gnArgHgt<\gnCornerHgt \gnArgHgt=0pt%
\else \advance \gnArgHgt by -\gnCornerHgt%
\fi \raise\gnArgHgt\hbox{$\ulcorner$} \box\gnBoxA %
\raise\gnArgHgt\hbox{$\urcorner$}}

\newcommand\Zero{\mathit{Zero}}
\newcommand\pZero{\mathit{pZero}}
\newcommand\sZero{\mathit{sZero}}
\newcommand\nZero{\mathit{nZero}}

\newcommand\One{\mathit{One}}
\newcommand\pOne{\mathit{pOne}}
\newcommand\sOne{\mathit{sOne}}
\newcommand\nOne{\mathit{nOne}}

\newcommand\Time{\mathit{Time}}

\newcommand\State{\mathit{State}}
\newcommand\pState{\mathit{pState}}
\newcommand\sState{\mathit{sState}}
\newcommand\nState{\mathit{nState}}

\newcommand\Tape{\mathit{Tape}}

\newcommand\Negative{\mathit{Negative}}
\newcommand\Accept{\mathit{Accept}}
\newcommand\true{\mathit{true}}
\newcommand\false{\mathit{false}}
  
\title{{\sc DValue} for Boolean games is EXP-complete}  
\author{Egor Ianovski}
  
\begin{document}

  \maketitle
  
  \begin{abstract}
    We show that the following problem is EXP-complete: given a rational $v$
    and a two player, zero-sum Boolean game $G$ determine whether the value of $G$
    is at least $v$. The proof is via a translation of the proof of the same
    result for Boolean circuit games in \cite{Feigenbaum1995}.
  \end{abstract}
  
  \section{Preliminaries}
  
    We will be using the encoding of \cite{Ianovski2O13} to replicate the proof of \cite{Feigenbaum1995}.
    A familiarity with \cite{Ianovski2O13} will make the proof much easier to follow.
    
    The specific breed of Boolean games that concerns us here has two players, and since their
    goals are purely conflicting only one goal formula is necessary.
  
    \begin{definition}
      A two player, zero-sum Boolean game consists of a set of variables, $\Phi$,
      partitioned into two sets, $\Phi_1$ and $\Phi_2$, as well as a propositional logic formula
      $\gamma_1$ over $\Phi$.
      
      The game is played by letting Player One choose a truth assignment to the variables in $\Phi_1$
      and Player Two to the variables in $\Phi_2$. If the resulting truth assignment satisfies $\gamma_1$,
      Player One wins the game. If it does not, Player Two wins the game.
    \end{definition}
    
    The algorithmic question of interest is a decision version of {\sc Value}.
    
    \algproblem{DValue}{\vspace{0.1cm} A two player zero-sum Boolean game $G$ and a rational $v$.}
  {\vspace{0.2cm} YES if the value of $G$ is at least $v$, NO otherwise.}

    \subsection{Talking about numbers}
    
      As we will be dealing with sequences of propositional variables a lot, we will use the notation
      $\overline{p_i}$ to mean $p_1,\dots,p_n$. The length of the sequence, $n$, will be clear from context.
    
      Sequences have order, which is essential to our approach of representing numbers via truth assignment.
      This is done in the standard way - the truth assignment to $\overline{p_i}$ that sets $p_i$ to $\true$
      is treated as the binary integer with the $i$th most significant bit set to 1. This leads to the following
      definition:
      
      \begin{definition}
      Let $\overline{p_i}$ be a sequence of $n$ propositional variables,
      and $\nu$ some truth assignment to $\overline{p_i}$.
      We use $\denote{\overline{p_i}}$ to denote the numeric value associated
      with $\nu$ via its assignment to $\overline{p_i}$.
    \end{definition}
    
     \begin{definition}
      Let $j$ be a binary integer in $[0,2^n-1]_\NN$. Let $j[i]=p_i$ if the $i$th
      bit of $j$ is 1 and $j[i]=\neg p_i$ otherwise. We use $\bs{\Godelnum{j}}$ to denote
      the formula asserting that $\denote{\overline{p_i}}=j$:
      $$\bs{\Godelnum{j}}=\bigwedge_{1\leq i\leq n}j[i].$$
      Note also that the size of $\bs{\Godelnum{j}}$ is linear in $|j|$.
      
      In the case where we wish to be precise as to which sequence $\bs{\Godelnum{j}}$
      is defined over, we use $\bs{\Godelnum{j}}({\overline{p_i}})$.
    \end{definition}
    
    We introduce a few formulae shorthands. These formulae take sequences of
    of variables as arguments. We assume throughout that every sequence passed to
    a formula is of the same length.
    \begin{lemma}[\cite{Ianovski2O13}]
      Let $\bs{Succ}(\overline{p_i},\overline{q_i})$ denote a term that is
      true if and only if $\denote{\overline{p_i}}+1=\denote{\overline{q_i}}$.
      
      $\bs{Succ}(\overline{p_i},\overline{q_i})$ can be replaced by a formula
      polynomial in $|\overline{p_i}|+|\overline{q_i}|$.
    \end{lemma}
    
    \begin{lemma}
      Let $\bs{Equal}(\overline{p_i},\overline{q_i})$ denote a term that is
      true if and only if $\denote{\overline{p_i}}=\denote{\overline{q_i}}$.
      
      $\bs{Equal}(\overline{p_i},\overline{q_i})$ can be replaced by a formula
      linear in $|\overline{p_i}|+|\overline{q_i}|$.
    \end{lemma}
    \begin{proof}
      Two binary integers are equal if and only if they are bitwise equal.
      This gives us the following:
      $$\bs{Equal}(\overline{p_i},\overline{q_i})=\bigwedge_{1\leq i\leq |\overline{p_i}|}(p_i\lra q_i).$$
    \end{proof}
    
    Of use in the next section is the fact that we can also deal with the less-than
    order.
    
    \begin{lemma}
	Let $\bs{Less}(\overline{p_i},\overline{q_i})$ denote a term that is true if
	and only if $\denote{\overline{p_i}}<\denote{\overline{q_i}}$.
	
	$\bs{Less}(\overline{p_i},\overline{q_i})$ can be replaced
	by a formula of propositional logic polynomial in $|\overline{p_i}|+|\overline{q_i}|$.
    \end{lemma}
    \begin{proof}
	Let $a[i]$ be the $i$th most significant bit of $a$. 
	
	Intuitively, if $\denote{\overline{p_i}}<\denote{\overline{q_i}}$ for two big-endian binary digits then there exists a
	$k$ such that:
	\begin{align*}
	  \denote{\overline{p_i}}[i]&=\denote{\overline{q_i}}[i],\quad i\leq k.\\
	  \denote{\overline{p_i}}[k+1]&=0,\quad \denote{\overline{q_i}}[k+1]=1.
	\end{align*}
	That is, the first bit where the two integers differ is a 1 for $\denote{\overline{q_i}}$
	and a 0 for $\denote{\overline{p_i}}$. This is clearly both necessary and sufficient.
	
	Since there are only $|\overline{p_i}|$ possible values of $k$, this can be replaced by
	a polynomial size formula that looks as follows:
	$$\bigvee_{0\leq k< |\overline{p_i}|} \bs{Succ}(p_1,\dots,p_k,q_1,\dots, q_k).$$
    \end{proof}
    \begin{lemma}
	Let $\bs{LessEq}(\overline{p_i},\overline{q_i})$ denote a term that is true if
	and only if $\denote{\overline{p_i}}\leq\denote{\overline{q_i}}$.
	
	$\bs{LessEq}(\overline{p_i},\overline{q_i})$ can be replaced
	by a formula of propositional logic polynomial in $|\overline{p_i}|+|\overline{q_i}|$.
    \end{lemma}
    \begin{lemma}
	Let $\bs{Add}(\overline{p_i},\overline{q_i},\overline{r_i})$ denote a term that is true if
	and only if $\denote{\overline{p_i}}+\denote{\overline{q_i}}=\denote{\overline{r_i}}$.
	
	$\bs{Add}(\overline{p_i},\overline{q_i},\overline{r_i})$ can be replaced
	by a formula of propositional logic polynomial in $|\overline{p_i}|+|\overline{q_i}|+|\overline{r_i}|$.
    \end{lemma}
    \begin{proof}
      We first have to ensure that $\denote{\overline{p_i}}+\denote{\overline{q_i}}$ is not too big, i.e. less
      than $2^n$. For this it is sufficient to rule out three cases: $p_1$ and $q_1$ being true; $p_1$, $p_2$ and $q_2$ being
      true; $q_1$, $p_2$ and $q_2$ being true.
      
      Having guaranteed this, the bitwise case for $i$ checks if there is a carry bit to account for (whether $p_{i+1}$ and $q_{i+1}$ are true)
      and handles the rest in the natural fashion.
      \begin{align*}
	\bs{Add}(\overline{p_i},\overline{q_i},\overline{r_i})=&\neg(p_1\wedge q_1)\wedge\neg\big( (p_1\vee q_1)\wedge p_2\wedge q_2\big)\\
	&\wedge\bigwedge_{i=1}^{n-1} \Big[\Big((p_{i+1}\wedge q_{i+1})\wedge\big( (p_i\wedge q_i\wedge r_i)\vee(p_i\wedge\neg q_i\wedge \neg r_i)\\
	&\vee(\neg p_i\wedge q_i\wedge \neg r_i)\vee(\neg p_i\wedge\neg q_i\wedge r_i) \big)\Big)\\
	&\vee\Big(\neg(p_{i+1}\wedge q_{i+1})\wedge\big( (p_i\wedge q_i\wedge \neg r_i)\vee(p_i\wedge\neg q_i\wedge r_i)\\
	&\vee(\neg p_i\wedge q_i\wedge r_i)\vee(\neg p_i\wedge\neg q_i\wedge \neg r_i) \big)\Big)\Big]\\
	&\wedge \big((p_n\wedge q_n\wedge \neg r_n)\vee(p_n\wedge\neg q_n\wedge r_n)\\
	&\vee(\neg p_n\wedge q_n\wedge r_n)\vee(\neg p_n\wedge\neg q_n\wedge \neg r_n)\big).
      \end{align*}
    \end{proof}
    
    \begin{lemma}
	Let $\bs{Sub}(\overline{p_i},\overline{q_i},\overline{r_i})$ denote a term that is true if
	and only if $\denote{\overline{p_i}}-\denote{\overline{q_i}}=\denote{\overline{r_i}}$.
	
	$\bs{Sub}(\overline{p_i},\overline{q_i},\overline{r_i})$ can be replaced
	by a formula of propositional logic polynomial in $|\overline{p_i}|+|\overline{q_i}|+|\overline{r_i}|$.
    \end{lemma}
    \begin{proof}
      Same idea as with addition.
    \end{proof}
    
    Finally, at times in lieu of testing $\denote{\overline{p_i}}$ for one of these
    three relations against $\denote{\overline{q_i}}$, we may wish to test,
    for example, whether $\denote{\overline{p_i}}<3$. One way to achieve
    this using the operations we have defined so far is:
    $$\bs{\Godelnum{3}}({\overline{q_i}})\wedge\bs{Less}(\overline{p_i},\overline{q_i}).$$
    While this works, it has the disadvantage of introducing the superfluous variables 
    $\overline{q_i}$. As the sort of games we use to encode Turing machines are already
    overburdened with variables, it is best not to introduce new ones unnecessarily.
    
    To handle this, we will abuse notation somewhat. We will not require that in a formula 
    with parameters such as $\bs{Less}(\overline{p_i},\overline{q_i})$ all the $p_i$
    and $q_i$ are propositional variables. We will also allow propositional constants,
    $\true$ and $\false$. This will allow us to express the desired relation as:
    $$\bs{Less}(\overline{p_i},\false,\dots,\false,\true,\true).$$
    
    To allow a more concise representation, we will adopt the convention that $\bs{\Godelnum{i}}$
    as a parameter to a formula should be read as a sequence of $\true$ and $\false$ encoding $i$,
    and not the conjunction testing whether $i$ holds in the given assignment.\footnote{It would, of course,
    be more correct to introduce a new notation for this contraction, but we feel that in this
    case overloading existing notation is more readable than introducing new symbols.}
    This will allow us to express the desired formula as:
    $$\bs{Less}(\overline{p_i},\bs{\Godelnum{3}}).$$
    
    Note that the following equivalence holds:
    $$\bs{\Godelnum{i}}(\overline{p_i})\Lra\bs{Equal}(\overline{p_i})(\bs{\Godelnum{i}}).$$
    
    \subsection{Additional shorthand}
    
    In \cite{Ianovski2O13} there was a $\bs{OneOf}$, and we can equally
      have a $\bs{NoneOf}$.
      \begin{lemma}
	Let $\bs{OneOf}(\overline{p_i})$ and $\bs{NoneOf}(\overline{p_i})$
	denote terms that are true just if exactly 1 and 0 respectively of
	$p_i$ are true. These terms are replaceable by formulae of propositional
	logic polynomial in $\overline{p_i}$.
      \end{lemma}
      Likewise, we could introduce an $\bs{nOf}$ for any constant $n$,
      but not $\bs{Of}(n)(\overline{p_i})$ with $n$ as a parameter.
      
      We shall not have any need of $\bs{nOf}$, but at times we do
      use an operation that could be called $\bs{MoreThanOneOf}$.
      However, in lieu of introducing new notation we will make use of
      the following equivalence:
      $$\bs{MoreThanOneOf}(\overline{p_i})\Lra\neg\bs{OneOf}(\overline{p_i})\wedge\neg\bs{NoneOf}(\overline{p_i}).$$
    
    \subsection{Games of any value}
    
      \begin{definition}\label{def:game}
	Let $\FG(v)$ denote a two player, zero-sum game with value $v$.
      \end{definition}
      
      \begin{lemma}\label{lem:game_value}
	For $v\in [0,1]_\QQ$, $\FG(v)$ has a Boolean form and the size
	of that form is polynomial in $|v|$.
      \end{lemma}
      \begin{proof}
	Let $v=a/b$. Consider the game where Player One selects two numbers $c_1,c_2\in[0,b]_\NN$
	with the property that $c_2-c_1\equiv a\mod b +1$. Player Two selects $d\in[0,b]_\NN$. The
	game is won by Player One if $c_2\geq d\geq c_1$, $d\geq c_1>c_2$ or $c_1>c_2\geq d$.
	
	Intuitively, Player One selects an interval of length $a$, that is allowed to loop around
	the end points, and Player Two tries to name a number outside that interval. The value of
	the game is $v$ as can be witnessed by the equilibrium where Player One randomises over every interval
	and Player Two over every number with equal probability.
	
	This game can be given a Boolean representation in the following way:
	
	\begin{align*}
	  \Phi_1=&\{p_1,\dots,p_n,q_1,\dots,q_n, s_1,\dots,s_n,t_1,\dots,t_n\}\\
	  \Phi_2=&\{r_1,\dots,r_n\}\\
	  \gamma_1=&\big(\bs{Sub}(\overline{q_i},\overline{p_i},a)\wedge 
	  \bs{LessEq}(\overline{q_i},b)\wedge \bs{LessEq}(\overline{r_i},\overline{q_i})\wedge \bs{LessEq}(\overline{p_i},\overline{r_i})\big)\\
	  &\vee\Big(\bs{Add}(\overline{s_i},\overline{t_i},a))\wedge\bs{Sub}(\overline{q_i},0,\overline{s_i})
	  \wedge\bs{Sub}(b,\overline{p_i},\overline{t_i})\\
	  &\wedge\big( \bs{LessEq}(\overline{r_i},\overline{q_i})\vee\bs{LessEq}(\overline{p_i},\overline{r_i}) \big)\Big)\\
	  &\vee\bs{Less}(b,\overline{r_i}).
	\end{align*}
	
	The interpretation is that $\denote{\overline{p_i}}=c_1,\denote{\overline{q_i}}=c_2,\denote{\overline{r_i}}=d$.
	The $s$ and $t$ variables come in to play if Player One wishes to play a looping interval - in which case $\denote{\overline{s_i}}$
	is the distance between $0$ and $c_2$, while $\denote{\overline{t_i}}$ is the distance between $c_1$ and $b$. These variables
	are added to give us a way to check that if Player One plays a looping interval, its length is still $a$.
	
	The last disjunct of $\gamma_1$ serves to award the game to One should Two name a $d$ outside of $[0,b]_\NN$. The first
	disjunct handles the non-looping case, i.e. where $c_2\geq c_1$, and the second disjunct the looping case.
      \end{proof}
  
  \section{Main result}
  
    \begin{theorem}
      {\sc DValue} is EXP-complete.
    \end{theorem}
    \begin{proof}
      To see that the problem is in EXP, expand the Boolean game into normal form and run the PTIME algorithm.
      
      To see that the problem is EXP-hard we will show that given the triple $(M, K, w)$, where $M$ is a deterministic
      Turing machine, $K$ a computation bound and $w$ an input word, we can construct a Boolean game $G$
      and a rational $v$ such that the value of $G$ is at least $v$ just if $M$ on input $w$ accepts in at most $K$
      steps.
      
      We will use $k$ for the size of $K$, i.e. $|K|=k$.
      
      The idea of the proof is to use the encoding of \cite{Ianovski2O13} to replicate the proof in \cite{Feigenbaum1995}.
      For the benefit of the reader, we will reproduce the main thrust of the proof in \cite{Feigenbaum1995}.
      
      We wish to associate with $M$ a set of Horn clauses $S$ over a set of propositional variables $P$ such that $M$ accepts
      $w$ in at most $K$ steps if and only if there exists an assignment to the variables in $P$ satisfying
      every clause in $S$.
      
      $P$ contains of propositions of the form $p[t,l,a]$ and $p[t,l,(s,a)]$. The intended interpretation of 
      $p[t,l,a]$ is that cell $l$ contains symbol $a$ at computation step $t$. Without loss of generality,
      we are working on a binary alphabet, so $a$ is 0, 1, or the blank tape symbol $\bot$. The intended interpretation of
      $p[t,l,(s,a)]$ is that, in addition to the above, the head is over cell $l$ and in state $s$.
      
      $S$ contains three types of clauses. The first type describe the initial configuration of the machine.
      These consist of $p[0,0,(q_1,w[1])]$, $p[0,i,w[i]]$ for $1\leq i<|w|$, $p[0,i,\bot]$ for $i\geq|w|$, 
      and $\neg p[0,x,y]$ for any $x,y$
      not conforming to the preceding types. The second type describe  the transition rules of the machine.
      These take the form:
      \begin{align*}
      (p[t,l-1,\sigma_1]\wedge p[t,l,\sigma_2]\wedge p[t,l+1,\sigma_3])&\ria p[t+1,l,\sigma]\\
      (p[t,l-1,\sigma_1]\wedge p[t,l,\sigma_2]\wedge p[t,l+1,\sigma_3])&\ria \neg p[t+1,l,\sigma']
      \end{align*}
      choosing appropriate values for $\sigma_1,\sigma_2,\sigma_3$ and $\sigma'\neq\sigma$.
      The last clause is $p[K-1,0,(q_f,0)]$, asserting that $M$ accepts at time $K$.\footnote{We can
      without loss of generality assume that $M$ will only accept in the first cell with 0 written
      on the tape.}
      
      For convenience, we will treat every negative clause as a clause with an antecedent of $\false$. That is,
      instead of $\neg p[0,x,y]$ and $(p_1\wedge p_2\wedge p_3)\ria\neg q$ we will have the clauses
      $p[0,x,y]\ria\false$ and $(p_1\wedge p_2\wedge p_3\wedge q)\ria\false$. This will mean a clause
      can have anywhere between 0 and 4 proposition in the tail - a true initial condition, a false
      initial condition, a positive boundary rule, a positive rule, a negative rule.
      
      The game defined in \cite{Feigenbaum1995} proceeds by letting Player One choose 
      $r\in P$ and Player Two an element $C\in S,C=\bigwedge p_i\ria q$. Letting $R\sub P$ be the set of variables made true
      in the unique run of $M$ on $w$. The payoff to One is as follows:
      $$
      H(r,C)=\begin{cases}
      1+\alpha,\quad&r=q\\
      -1+\alpha,\quad&r=p_i\\
      \alpha,\quad&\text{otherwise.}\\
      \end{cases}
      $$
      In the above, $\alpha=\frac{j-1}{|R|}$, where $0\leq j\leq 4$ is the number of literals
      in the antecedent of $C$. In this framework the authors prove that the value of the game
      is $\geq 0$ if and only if $M$ accepts $w$ in at most $K$ steps.
      
      To simplify matters we assume that $|R|=2^{2k}$, i.e. we consider the first $2^k$ computation
      steps and $2^k$ tape cells. This can be done by endowing the machine with a ``do nothing"
      transition as in \cite{Ianovski2O13}.
      
      This is where we seek to hijack the rest of their proof. If we can construct
      a Boolean game that meets the same criteria described above, we are done. Before we do that, however,
      we must first normalise the payoffs to reflect the fact that the value of a Boolean game
      is necessarily in $[0,1]$. As such, it is clear that the argument of \cite{Feigenbaum1995}
      equally proves that if the game had the following payoffs:
      
      $$
      H'(r,C)=\begin{cases}
      3/4+\alpha/4,\quad&r=q\\
      1/4+\alpha/4,\quad&r=p_i\\
      1/2+\alpha/4,\quad&\text{otherwise.}\\
      \end{cases}
      $$
      then the value of the game is $\geq 1/2$ just if $M$ accepts $w$ in at most $K$ steps.\footnote{Because
      the utilities of the new game are obtained via the affine transformation $x/4+1/2$, and the
      equilibria of a finite game are invariant under affine transformations of utility.}
      
      Note that now the payoffs are within $[\frac{1}{4}-\frac{1}{4\cdot 2^{2k}},\frac{3}{4}+\frac{3}{4\cdot 2^{2k}}]$,
      and thus for $k\geq 1$ they are contained in the feasible range for a Boolean game, $[0,1]$.
      
      Now, suppose we can find a partition of a set of variables $\Phi'=\Phi_1'\uplus\Phi_2'$ and
      fifteen (polynomial size) formulae $\varphi_{r=q}^{j}$, $\varphi_{r=p_i}^j$, $\varphi_{\neq}^{j}$, $j\in\{0,1,2,3,4\}$,
      with the following properties:
      \begin{itemize}
	\item
	  Every truth assignment to $\Phi_1$ corresponds to a choice of $r\in P$.
	\item
	  Every truth assignment to $\Phi_2$ corresponds to a choice of $C\in S$.
	\item
	  $\varphi_{r=q}^j$ is true if and only if $C$ has $j$ elements in the tail
	  and $r$ is equal to the head of $C$. Mutatis mutandis, for the other $\varphi$.
      \end{itemize}
      We claim that at that point we are done. The following is the desired game:
      \begin{align*}
	\Phi&=\Phi'\cup G.\\
	\Phi_1&=\Phi_1'\cup G_1.\\
	\Phi_2&=\Phi_2'\cup G_2.\\
	\gamma_1&=\bigvee_{j}(\varphi_{r=q}^j\wedge\gamma_1^{3/4+\alpha/4})\vee\bigvee_{j}
	(\varphi_{r=p_i}^j\wedge\gamma_1^{1/4+\alpha/4})\vee\bigvee_{j}(\varphi_{\neq}^j\wedge\gamma_1^{1/2+\alpha/4}).
      \end{align*}
	$G=G_1\uplus G_2$ is the union of the (mutually disjoint) sets of variables necessary to play 
	$\FG(v)$ from Definition~\ref{def:game} for $v\in\{3/4+\alpha/4,1/4+\alpha/4,1/2+\alpha/4\}$ (as $\alpha$ varies, there are fifteen games
	in total). The $\gamma$s are the goal formulae of those games.
	
	By Lemma~\ref{lem:game_value}, these subgames can be constructed in polynomial time.
	
	As such, all that remains is to provide $\Phi'=\Phi_1'\uplus\Phi_2'$, $\varphi_{r=q}^{j}$, 
	$\varphi_{r=p_i}^j$, $\varphi_{\neq}^{j}$, $j\in\{0,1,2,3,4\}$.
	
	We start with Player One:
	\begin{align*}
	  \Phi_1'=&\{\mathit{Zero}_1,\mathit{One}_1\}\cup\{\mathit{Time}_1^i\}_{1\leq i\leq k}\cup\{\mathit{Tape}_1^i\}_{1\leq i\leq k}\\
	  &\cup\{\mathit{State}_1^i\}_{1\leq i\leq |Q|}.
	\end{align*}
	We map a truth assignment to $\Phi_1'$ to $r\in P$ in the following way:
	\begin{itemize}
	  \item
	    If both $\Zero_1$ and $\One_1$ are true, or more than one of $\{\mathit{State}_1^i\}_{1\leq i\leq |Q|}$
	    is true, then the assignment is treated as $p[0,0,0]$.\footnote{In \cite{Ianovski2O13} we punished a player for
	    making an illegal move by having them lose the game. However, we cannot do this here as if both players play
	    illegally the game would fail to be zero-sum. Instead, we pick a legal move for them.}
	  \item
	    If $\State_1^m$ and WLOG $\Zero_1$ is true then the assignment is treated 
	    as $p[\denote{\overline{\Time_1^i}},\denote{\overline{\Tape_1^i}},(q_m,0)]$.
	  \item
	    If all the state variables are false and WLOG $\Zero_1$ is true then the assignment is treated 
	    as $p[\denote{\overline{\Time_1^i}},\denote{\overline{\Tape_1^i}},0]$.
	\end{itemize}
	
	For Player Two we have a larger set of variables:
	\begin{align*}
	  \Phi_2'=&\{\pZero_2,\pOne_2,\Zero_2,\One_2,\sZero_2,\sOne_2,\nZero_2,\nOne_2\}\\
	  &\cup\{\mathit{Time}_2^i\}_{1\leq i\leq k}\cup\{\mathit{Tape}_2^i\}_{1\leq i\leq k}\cup\{\nState_2^i\}_{1\leq i\leq |Q|}\\
	  &\cup\{\pState_2^i\}_{1\leq i\leq |Q|}\cup\{\State_2^i\}_{1\leq i\leq |Q|}\cup\{\sState_2^i\}_{1\leq i\leq |Q|}\\
	  &\cup\{\mathit{Negative},\mathit{Accept}\}.
	\end{align*}
	We map a truth assignment to $\Phi_2'$ to $C\in S$ in the following way:
	\begin{itemize}
	  \item
	    An illegal configuration is mapped to $p[K-1,0,(q_f,0)]$.
	  \item
	    The $\overline{\Time_2^i}$, $\overline{\Tape_2^i}$ refer to the
	    cell/step specified by the consequent. Thus, if $\denote{\overline{\Time_2^i}}=0$,
	    the clause is treated as an initial configuration clause,
	    $p[0,\denote{\overline{\Tape_2^i}},w[i]]$.
	    If $\denote{\overline{\Tape_2^i}}$ is     
	    $0$ or $2^k-1$, then the clause is a boundary case and hence has only
	    two propositions in the tail. If $\mathit{Negative}$ is set to true,
	    then the assignment is mapped to a negative clause. With this in mind, the contents of
	    $\bigwedge p_i\ria q$ are derived from the assignment in the natural way.\footnote{Natural
	    to a reader who is familiar with \cite{Ianovski2O13}: recall, $\nOne$ refers to the contents of the \emph{n}ext
	    computation step, or the head of the clause. $\sOne$ and $\pOne$ are the \emph{s}uccessor
	    and \emph{p}redecessor of the central literal in the tail, and hence refer to the right and left cell.}
	  \item
	    $\mathit{Accept}$ is a special variable used to mark the fact that Player
	    Two is playing  $p[K-1,0,(q_f,0)]$. If $\mathit{Accept}$ is set to true,
	    and Player Two plays $\denote{\overline{\Time_2^i}}=K-1$, $\denote{\overline{\Tape_2^i}}=0$,
	    $\nZero_2$ and $\nState_2^{accept}$, the assignment is treated as
	    $p[K-1,0,(q_f,0)]$.
	    \footnote{This is technically redundant: Player Two could specify the accepting
	    clause by playing any illegal assignment, but we do not wish to make illegal play
	    a necessary aspect of the game.}
	\end{itemize}
	At this point the reader should convince themselves that the mapping defined above
	does, in fact, allow Player One to specify any proposition in $P$ and Player Two
	any clause in $S$.
	
	Let us now turn to $\varphi_{r=q}^{j}$. We will deal with $j=0$ and $j=3$. The case
	of $j=2$ is obtained from $j=3$ by changing the appropriate cell index and $j=1,4$
	is simply $\mathit{false}$, as Player One is incapable of guessing the consequent in that
	instance.
	
	For $j=0$ there are three possibilities to consider. Player Two may have correctly
	specified a positive initial condition, 
	the accepting clause, or played an illegal configuration. Recall that a negative initial condition clause
	is treated as $q\ria\false$, and thus falls under $j=1$.
	We also need
	not consider Player One playing an illegal configuration, as that cannot appear
	in the head of any clause and hence cannot satisfy $\varphi_{r=q}^{j}$.
	$$\varphi_{r=q}^{0}=\mathit{Init}\vee\mathit{Final}\vee\mathit{Illegal}_i.$$
	$\mathit{Init}$ requires that $\overline{\Time_2^i}$ encodes 0; the state variables
	are false unless $\overline{\Tape_2^i}$ encodes 0, in which case only $\State_2^1$
	is true; and if $\overline{\Tape_2^i}$ encodes $j$ then the $\nZero_2,\nOne_2$
	variables are played in accordance with $w[j]$. $\Negative$ and $\Accept$ are both
	false. Player One plays his time, tape variables such that they encode the
	same numbers as Player Two's, and likewise agrees on the state and content
	variables.

	$\mathit{Init}$ can be broken down into a correctness and a matching requirement.
	$$\mathit{Init}=\mathit{Init}_c\wedge\mathit{MatchHead}.$$
	Line by line, the formula below reads: if the chosen cell is not 0, the head is not
	over the cell. If the chosen cell is 0, the head is over the cell and in state $q_0$.
	If the chosen cell is $i< |w|$, then Player Two sets $w[i]\in\{\nZero_2\wedge\neg\nOne_2,\nOne_2\wedge\neg\nZero_2\}$
	to $\true$, depending on the bit of $w$. If the chosen cell is $i\geq |w|$, then
	the cell is blank. As the clause is neither accepting nor negative, both those
	variables are set to $\false$.
	\begin{align*}
	  \mathit{Init}_c=& \big(\neg\bs{\Godelnum{0}}(\overline{\Tape_2^i})\ria\bs{NoneOf}(\overline{\nState_2^i})\big)\\
	  &\wedge \Big(\bs{\Godelnum{0}}(\overline{\Tape_2^i})\ria\big(\nState_2^{initial}\wedge\bs{OneOf}(\overline{\nState_2^i})\big)\Big)\\
	  &\wedge\bigwedge_{0\leq i< |w|}\big( \bs{\Godelnum{i}}(\overline{\Tape_2^i})\ria w[i] \big)\\
	  &\wedge \big(\neg\bs{Less}(\overline{\Tape_2^i},\bs{\Godelnum{|w|}})\ria (\neg\nZero_2\wedge\neg\nOne_2)\big)\\
	  &\wedge\neg\Accept\wedge\neg\Negative.
	\end{align*}
	Note that the third line expands into $|w|$ conjuncts, so the formula is polynomial size.
	
	$\mathit{MatchHead}$ states that Player One specifies the same proposition as is in the head of
	Player Two's clause. That is, the cell is the same, the computation step is the same, the
	tape contents are the same and the machine state is the same. This is a general term that
	we will reuse in other subformulae.
	\begin{align*}
	  \mathit{MatchHead}=&\bs{Equal}(\overline{\Tape_1^i},\overline{\Tape_2^i})\wedge\bs{Equal}(\overline{\Time_1^i},\overline{\Time_2^i})\\
	  &\wedge(\Zero_1\lra\nZero_2)\wedge(\One_1\lra\nOne_2)\\
	  &\wedge\bigwedge_{1\leq i\leq |Q|}(\State_1^i\lra\nState_2^i).
	\end{align*}

	$\mathit{Final}$ likewise has a correctness and a matching requirement. The
	correctness requirement asks that Player Two set $\mathit{Accept}$ to $\true$
	and specifies $p[K-1,0,(q_f,0)]$. The matching requirement we can reuse
	from the preceding case.
	$$\mathit{Final}=\mathit{Final}_c\wedge\mathit{MatchHead}.$$
	\begin{align*}
	  &\mathit{Final}_c=\Accept\wedge\neg\Negative\wedge\nState_2^{accept}\wedge\bs{OneOf}(\overline{\nState_2^i})\\
	  &\wedge\bs{\Godelnum{K-1}}(\overline{\Time_2^i})\wedge\bs{\Godelnum{0}}(\overline{\Tape_2^i})\wedge\nZero_2\wedge\neg\nOne_2.
	\end{align*}
	
	$\mathit{Illegal}_i$ says that Player Two names an illegal configuration and
	Player One names $p[K-1,0,(q_f,0)]$.
	$$\mathit{Illegal}_i=\mathit{TwoIllegal}\wedge\mathit{OneFinal}.$$
	
	Let us list everything that could constitute an illegal assignment
	for Player Two:
	\begin{enumerate}
	  \item
	    The presence of both 1 and 0 in any specified cell.
	  \item
	    The presence of more than one state in any cell.
	  \item
	    The presence of the head in more than one cell in the tail.
	  \item
	    Player Two names computation step 0, but supplies an incorrect initial configuration of the machine.
	  \item
	    Player Two plays $\Accept$ and does not correctly describe $p[K-1,0,(q_f,0)]$.
	  \item
	    Player Two names computation step $\geq 1$ and supplies a clause
	    inconsistent with the transition rules of the machine.
	\end{enumerate}
	We will introduce a formula for each item. $\mathit{TwoIllegal}$ will be the disjunction
	of these formulae.
	\begin{align*}
	  1=&(\pZero_2\wedge\pOne_2)\vee(\Zero_2\wedge\One_2)\vee(\sZero_2\wedge\sOne_2)\\
	  &\vee(\nZero_2\wedge\nOne_2).
	\end{align*}
	\begin{align*}
	  2=&\big(\neg\bs{OneOf}(\overline{\nState_2^i})\wedge\neg\bs{NoneOf}(\overline{\nState_2^i})\big)\\
	  &\vee\big(\neg\bs{OneOf}(\overline{\pState_2^i})\wedge\neg\bs{NoneOf}(\overline{\pState_2^i})\big)\\
	  &\vee\big(\neg\bs{OneOf}(\overline{\State_2^i})\wedge\neg\bs{NoneOf}(\overline{\State_2^i})\big)\\
	  &\vee\big(\neg\bs{OneOf}(\overline{\pState_2^i})\wedge\neg\bs{NoneOf}(\overline{\sState_2^i})\big).\\
	\end{align*}
	\begin{align*}
	  3=&\big(\neg\bs{NoneOf}(\overline{\pState_2^i})\wedge\neg\bs{NoneOf}(\overline{\State_2^i})\big)\\
	  &\vee\big(\neg\bs{NoneOf}(\overline{\pState_2^i})\wedge\neg\bs{NoneOf}(\overline{\sState_2^i})\big)\\
	  &\vee\big(\neg\bs{NoneOf}(\overline{\State_2^i})\wedge\neg\bs{NoneOf}(\overline{\sState_2^i})\big).\\
	\end{align*}
	\begin{align*}
	  4=&\bs{\Godelnum{0}}(\overline{\Time_2^i})\wedge\Big(\Big(\neg\Negative\wedge\big(\bigvee_{0\leq i< |w|}( \bs{\Godelnum{i}}(\overline{\Tape_2^i})\wedge\neg w[i] )\\
	  &\vee    (\neg\bs{Less}(\overline{\Tape_2^i},\bs{\Godelnum{|w|}})\wedge(\nZero_2\vee\nOne_2))\big)\Big)\\
	  &\vee \Big(\Negative\wedge\big(\bigvee_{0\leq i< |w|}( \bs{\Godelnum{i}}(\overline{\Tape_2^i})\wedge w[i] )\\
	  &\vee    (\neg\bs{Less}(\overline{\Tape_2^i},\bs{\Godelnum{|w|}})\wedge(\neg\nZero_2\wedge\neg\nOne_2))\big)\Big) \Big)
	\end{align*}
	\begin{align*}
	  5&=\Accept\\
	  &\wedge\big(\neg\bs{\Godelnum{K-1}}(\overline{\Time_2^i})\vee\neg\Godelnum{0}(\overline{\Tape_2^i})\vee\neg\nState_2^{accept}\vee\neg\nZero_2\big)
	\end{align*}
	The last formula we will not provide in its entirety. Its general form is a disjunction:
	$$6=\neg\bigvee_{Rule\in M}Rule.$$
	That is, we check whether any rule is consistent with the clause. A difficulty arises
	because a rule of the form $(q_i,\sigma)\ria(q_j,D,\sigma')$ gives rise to as much
	as 24 different Horn clauses - boundary cases and locations of the head.
	Of course 24 is a constant, so as far as our proof goes there is no problem in introducing
	that many terms into the disjunction for every rule of the machine, but unfortunately
	this document is too narrow to contain such a truly marvellous proof. Instead we will
	give a concrete example of one specific case: the rule $(q_3,0)\ria(q_4,R,1)$ where
	the head is initially in the middle cell and the middle
	is neither 0 nor $2^k-1$.
	\begin{align*}
	  \neg\bs{\Godelnum{0}}(\overline{\Tape_2^i})\wedge\neg\bs{\Godelnum{2^k-1}}
	  (\overline{\Tape_2^i})\wedge\neg\bs{\Godelnum{0}}(\overline{\Time_2^i})\\
	  \wedge\Zero_2\wedge\nOne_2\wedge\State_2^3\wedge\bs{NoneOf}(\nState_2^i).
	\end{align*}
	We will also need to introduce `negative rules' to correspond to what the
	machine does not do. These will be treated in a similar way, all that needs
	to be mentioned is that for each $(q_i,\sigma)\ria(q_j,D,\sigma')$
	there will be only polynomially ($O(|Q|\cdot2\cdot|\Sigma|)$) many $(q_i,\sigma)\ria\neg(q_j',D',\tau)$.

	So much for $j=0$. Let us turn to $j=3$.
	
	This turns out to be a lot easier as we have already done much of the gruntwork. All we
	need is for Player Two to name a step $\geq 1$, a cell $\geq 1$ and $<2^k-1$, a correct configuration, and for
	Player One to guess the head.
	\begin{align*}
	\varphi_{r=q}^{3}=&\neg\bs{\Godelnum{0}}(\overline{\Time_2^i})\wedge\neg\bs{\Godelnum{0}}
	(\overline{\Tape_2^i})\wedge\neg\bs{\Godelnum{2^k-1}}(\overline{\Tape_2^i})\\
	&\wedge\neg\mathit{TwoIllegal}\wedge\mathit{MatchHead}.
	\end{align*}
	
	Next up is $\varphi_{r=p_i}^j$. We will deal with $j=4$. There is no case for $j=0$, and $j=1$, $j=2$,
	$j=3$ can be easily obtained from $j=4$.
	
	Let us start by introducing the formulae checking for Player One guessing the tail. We have already
	seen $\mathit{MatchHead}$, which is applicable in the case of $j=4$ as we treat negative clauses
	as $\bigwedge p_i\wedge q\ria\false$. The others are built similarly.
	\begin{align*}
	  \mathit{MatchLeft}=&\bs{Succ}(\overline{\Tape_2^i},\overline{\Tape_1^i})\wedge\bs{Succ}(\overline{\Time_1^i},\overline{\Time_2^i})\\
	  &\wedge(\Zero_1\lra\pZero_2)\wedge(\One_1\lra\pOne_2)\\
	  &\wedge\bigwedge_{1\leq i\leq |Q|}(\State_1^i\lra\pState_2^i).
	\end{align*}
	\begin{align*}
	  \mathit{MatchCentre}=&\bs{Equal}(\overline{\Tape_1^i},\overline{\Tape_2^i})\wedge\bs{Succ}(\overline{\Time_1^i},\overline{\Time_2^i})\\
	  &\wedge(\Zero_1\lra\Zero_2)\wedge(\One_1\lra\One_2)\\
	  &\wedge\bigwedge_{1\leq i\leq |Q|}(\State_1^i\lra\State_2^i).
	\end{align*}
	\begin{align*}
	  \mathit{MatchRight}=&\bs{Succ}(\overline{\Tape_1^i},\overline{\Tape_2^i})\wedge\bs{Succ}(\overline{\Time_1^i},\overline{\Time_2^i})\\
	  &\wedge(\Zero_1\lra\sZero_2)\wedge(\One_1\lra\sOne_2)\\
	  &\wedge\bigwedge_{1\leq i\leq |Q|}(\State_1^i\lra\sState_2^i).
	\end{align*}
	Note that this is all we need to capture the case where Player One plays legally:
	\begin{align*}
	\varphi_{r=p_i}^4=&\mathit{Illegal}_{r=p_i}\\
	&\vee\Big(\neg\bs{\Godelnum{0}}(\overline{\Time_2^i})\wedge\neg\bs{\Godelnum{0}}
	(\overline{\Tape_2^i})\wedge\neg\bs{\Godelnum{2^k-1}}(\overline{\Tape_2^i})\\
	&\wedge\neg\mathit{TwoIllegal}\wedge\Negative\\
	&\wedge(\mathit{MatchHead}\vee\mathit{MatchLeft}\vee\mathit{MatchCentre}\vee\mathit{MatchRight})\Big).
	\end{align*}
	$\mathit{Illegal}_{r=p_i}$ is also relatively simple. Player One must make a violation, and Player Two
	needs to play a legal clause with $p[0,0,0]$ in the tail.
	\begin{align*}
	  \mathit{Illegal}_{r=p_i}=&\mathit{OneIllegal}\wedge\neg\mathit{TwoIllegal}\\
	  &\wedge
	  \Big( \big(\pZero_2\wedge\bs{NoneOf}(\overline{\pState_2^i})\wedge\bs{\Godelnum{1}}(\overline{\Time_2^i})\wedge\bs{\Godelnum{1}}(\overline{\Tape_2^i}) \big)\\
	  &\vee\big(\Zero_2\wedge\bs{NoneOf}(\overline{\State_2^i})\wedge\bs{\Godelnum{1}}(\overline{\Time_2^i})\wedge\bs{\Godelnum{0}}(\overline{\Tape_2^i})\big)  \Big)
	\end{align*}
	Player One does not have a lot of creativity in how to play incorrectly:
	\begin{align*}
	\mathit{OneIllegal}=&(\One_1\wedge\Zero_1)\\
	&\vee\big(\neg\bs{OneOf}(\overline{\State_1^i})\wedge\neg\bs{NoneOf}(\overline{\State_1^i})\big).
	\end{align*}
      
	Finally we come to $\varphi_{\neq}^{j}$, where we look at $j=3$.
	
	There are four cases: both players play correctly and diverge. Player One plays
	incorrectly and Player Two plays a correct clause not covering step/cell $(0,0)$.
	Player Two plays incorrectly and Player One plays a correct proposition not
	covering the step/cell $(K-1,0)$, and of course both players could play
	incorrectly, in which case $p[0,0,0]$ does not cover $p[K-1,0,(q_f,0)]$.
	\begin{align*}
	\varphi_{\neq}^{3}=&\neg\Negative\wedge\neg\Accept\\
	&\wedge\neg\bs{\Godelnum{0}}(\overline{\Time_2^i})\wedge\neg\bs{\Godelnum{0}}(\overline{\Tape_2^i})\wedge\neg\bs{\Godelnum{2^k-1}}(\overline{\Tape_2^i})\\
	&\wedge\big(\mathit{BothCorrect}\vee\mathit{TwoCorrect}
	\vee\mathit{OneCorrect}\vee\mathit{NoneCorrect}\big).
	\end{align*}
	We already have all the tools we need.
	\begin{align*}
	  \mathit{Both}&\mathit{Correct}=\neg\mathit{OneIllegal}\wedge\neg\mathit{TwoIllegal}\\
	  &\wedge\neg(\mathit{MatchHead}\vee\mathit{MatchLeft}\vee\mathit{MatchCentre}\vee\mathit{MatchRight}).
	\end{align*}
	\begin{align*}
	  \mathit{Two}\mathit{Correct}=&\mathit{OneIllegal}\wedge\neg\mathit{TwoIllegal}\\
	  &\wedge\neg\big(\bs{\Godelnum{1}}(\overline{\Time_2^i})\wedge\bs{\Godelnum{0}}(\overline{\Tape_2^i})\wedge\Zero_2\big)\\
	  &\wedge\neg\big(\bs{\Godelnum{1}}(\overline{\Time_2^i})\wedge\bs{\Godelnum{1}}(\overline{\Tape_2^i})\wedge\pZero_2\big).
	\end{align*}
	\begin{align*}
	  \mathit{One}&\mathit{Correct}=\neg\mathit{OneIllegal}\wedge\mathit{TwoIllegal}\\
	  &\wedge\neg\big(\bs{\Godelnum{K-1}}(\overline{\Time_1^i})\wedge\bs{\Godelnum{0}}(\overline{\Tape_1^i})\wedge\Zero_1\wedge\State_1^{accept}\big).
	\end{align*}
	\begin{align*}
	  \mathit{None}&\mathit{Correct}=\mathit{OneIllegal}\wedge\mathit{TwoIllegal}.\\
	\end{align*}

      \end{proof}
      
      \section{Future directions}
      
	{\sc Value} is both the more natural and the more interesting algorithmic problem than its
	decision counterpart, and hence deserves investigations.
	
	Difficulties can be anticipated because superpolynomial function classes are not
	well understood. The convenient self-reducibility taken for granted in NP no
	longer applies. See, for example, \cite{Impagliazzo1989}.
	
	Seeing how in the case of circuit games, too, next to nothing is known about the
	complexity of function problems answering this question could lead to a range
	of new results about succinctly represented games.

  \bibliography{references}
  \bibliographystyle{plain}
\end{document}